\documentclass[conference]{IEEEtran}

\usepackage{cite}
\usepackage{textcomp}
\usepackage{amsmath, amssymb, amsfonts}
\usepackage{accents}
\usepackage{color}
\usepackage{graphicx}
\usepackage{enumerate}
\usepackage{booktabs}
\usepackage{subcaption}
\usepackage{mathrsfs}
\usepackage{mathtools}
\usepackage{dsfont}
\usepackage{relsize}
\usepackage{cmtt}

\usepackage{algorithm}
\usepackage{algpseudocode}

\usepackage[dvipsnames]{xcolor}

\makeatletter
\let\NAT@parse\undefined
\makeatother
\usepackage{hyperref}  
\hypersetup{
    colorlinks= true,
    citecolor= black,
    linkcolor= black, 
    pdfborder= {0 0 0},
}

\DeclareMathOperator{\minimize}{minimize}

\DeclareMathOperator{\E}{\mathsf{E}}

\DeclareMathOperator{\sense}{\mathsf{sense}}
\DeclareMathOperator{\comm}{\mathsf{comm}}

\DeclareMathOperator{\fail}{\mathsf{fail}}
\DeclareMathOperator{\success}{\mathsf{succ}}

\newtheorem{lemma}{Lemma}

\newtheorem{theorem}{Theorem}

\newcounter{parnum}

\begin{document}

\title{Status Updating via Integrated Sensing and Communication: Freshness Optimisation}
\author{
Touraj Soleymani, Mohamad Assaad, and John S. Baras
\thanks{A preliminary version of this work has been submitted to the IEEE SPAWC 2026 conference. T.~Soleymani is with the University of London, United Kingdom (e-mail: {\tt\footnotesize touraj.soleymani@citystgeorges.ac.uk}). M.~Assaad is with the CentraleSup\'{e}lec, University of Paris-Saclay, France (e-mail: {\tt\footnotesize mohamad.assaad@centralesupelec.fr}). J.~S.~Baras is with the University of Maryland College Park, United States (e-mail: {\tt\footnotesize baras@umd.edu}).}%
}
\maketitle

\begin{abstract}
In this paper, we study how sensing and communication should be jointly coordinated in integrated sensing and communication (ISAC) systems to maintain timely situational awareness under reliability and resource constraints. We consider an ISAC-enabled base station that supports a remote source by dynamically choosing between sensing new state information and communicating previously acquired information, with the two operations semantically intertwined rather than serving separate targets and users. Both sensing and communication are unreliable and costly. The objective is to optimise a long-term cost that captures information freshness at the source, measured by the age of information (AoI), together with sensing and communication overheads. The resulting sequential decision problem is formulated as an infinite-horizon Markov decision process (MDP) with two-dimensional AoI states that capture information freshness at the source and at the base station. We prove that the optimal stationary policy admits a monotone threshold structure characterised by a nondecreasing switching curve in the AoI state space, and show that, as the base-station information becomes staler, the system increasingly favours sensing over communication. Our numerical analysis corroborates the theoretical findings.
\end{abstract}

\begin{IEEEkeywords}
Age of information, Markov decision processes, semantic communication, remote navigation.
\end{IEEEkeywords}

\section{Introduction}
Integrated sensing and communication (ISAC) has recently emerged as a foundational paradigm for next-generation wireless systems~\cite{Liu2022JSAC, liu2020joint, luong2020resource}. This integration is particularly promising for tasks built upon remote navigation, where an agent makes motion decisions based on externally acquired navigational information, rather than solely on its onboard perception. Examples include unmanned vehicles requiring external support for situational awareness and collision avoidance. In such systems, accurate navigation hinges on the timely delivery of relevant state information, which must be both sensed and communicated under stringent reliability and resource constraints. Consequently, performance in these systems is fundamentally constrained not only by communication latency, but also by the freshness of the information available at the decision maker. Even when latency is small, outdated state information can lead to unsafe or inefficient actions, especially in highly dynamic environments. This observation motivates the use of semantic metrics that explicitly capture the information freshness within the sensing-communication-control loop in the system-level design process. A simple yet powerful metric in this context is the age of information (AoI), defined as the elapsed time since the most recent state information was successfully acquired~\cite{kaul2012, yates2021age}.

Note that ISAC can be studied at several levels, including waveform-level integration through joint signalling design, resource-level integration through shared resource design, and decision-level integration through unified operational policy design. This paper mainly focuses on integration at the decision level, while abstracting the underlying physical interaction through sensing and communication reliability and resource parameters. More specifically, we consider an ISAC-enabled base station that supports a remote source by dynamically choosing between sensing new state information and communicating previously acquired information, with the two operations semantically intertwined rather than serving separate targets and users. Both sensing and communication are unreliable and costly. The objective is to optimise a long-term cost that captures information freshness at the source, measured by the AoI, together with sensing and communication overheads. We study a setting that captures the fundamental sensing-communication tradeoff and serves as a rigorous analytical foundation for understanding the dynamics of timely information acquisition.

\subsection{Related Work}
This work lies at the intersection of information freshness and ISAC. We review the related literature below.

\subsubsection{Semantics of Information} The AoI was introduced as a semantic metric for quantifying the timeliness of status updates in real-time systems~\cite{kaul2012, yates2021age, kosta2017, sun2022age}. Subsequent work has studied the AoI across different network settings, including random delay~\cite{sun2019}, energy harvesting~\cite{arafa2019age}, retransmissions~\cite{ceran2019average}, gossip networks~\cite{buyukates2022vers}, broadcast channels~\cite{kadota2018scheduling}, medium-access channels~\cite{chen2022age}, and massive-scale networks~\cite{liyan2025cyclic}. In these models, freshness at the receiver is governed by a single information-update mechanism. In contrast, our setting couples two distinct operations: sensing first refreshes the information stored at the base station, and communication subsequently delivers that information to the corresponding remote source. This separation gives rise to a two-dimensional freshness state and a fundamental tradeoff between sensing new state information and communicating previously acquired information. Beyond the AoI, richer notions of semantics have been considered in the literature~\cite{uysal2022semantic}, e.g., the age of incorrect information (AoII) and the value of information (VoI). The AoII was proposed as a semantic metric quantifying the time elapsed since the receiver's estimate was last correct~\cite{maatouk2020age, maatouk2022age, chen2024minimizing}. The VoI was introduced as a decision-theoretic measure quantifying the difference between the benefit and the cost of new measurements in networked control systems~\cite{mywodespaper, touraj-thesis, voi, voi2, soleymani2024found}, and has been analysed under packet loss and time delay~\cite{erasure2023, soleymani2024consis} and related to the AoI~\cite{soleymaniCUP}. However, frameworks based on the AoII and the VoI require an explicit rule for estimating the belief of each remote source, making the analysis dependent on the specific dynamics of the state process. Instead, we adopt the AoI, which abstracts away the underlying state process and therefore yields a more general framework applicable across time-sensitive systems.

\subsubsection{Integrated Sensing and Communication} ISAC has recently been identified by the International Telecommunication Union as one of the six usage scenarios for 6G~\cite{wp5d2023draft}. Its emergence is driven by spectrum scarcity, shared sensing-communication hardware, and the growing role of perceptive networks in which sensing is treated as a primary service~\cite{Liu2022JSAC, liu2020joint, luong2020resource}. Research on ISAC can be organised according to the level at which sensing and communication are integrated. At the waveform level, where integration is achieved by designing the dual-functional signal~\cite{mishra2019toward, hua2023mimo, xiong2023fundamental, sturm2011waveform, liu2021cramer}, the design variables include waveforms, transmit precoding and beamforming matrices, and spatial beampatterns, while the objective is to balance sensing signal-quality metrics against those for communication. At the resource level, where integration is achieved by partitioning a shared, finite resource budget between the two functions~\cite{xu2022robust, he2023full, shi2024beamforming, liao2024power, peng2024traj, liu2024uav_iot_isac}, the design variables include transmit power, subcarrier assignments, and antenna or subarray partitions, while the objective is to ensure feasibility and efficiency under shared-resource contention. A more task-oriented perspective is decision-level integration, which treats sensing and communication as actions to support downstream objectives such as monitoring, navigation, and control. However, much of the ISAC literature remains centred on traditional performance metrics such as detection probability, achievable rate, the Cram\'er--Rao bound (CRB), and signal-to-interference-plus-noise ratio (SINR). These metrics are fundamental for physical-layer design, but they do not directly capture the freshness requirements of time-sensitive closed-loop applications.

\subsubsection{Timely Sensing and Communication} Recent works have begun to study information freshness in ISAC systems~\cite{jia2026study, ge2026age, bai2025age, mei2025aoi, hu2026balancing, liu2025joint}. More specifically, Jia~\emph{et al.}~\cite{jia2026study} proposed a timeliness-centric ISAC framework in which sensing AoI and communication delay are jointly optimised through dynamic programming. Ge~\emph{et al.}~\cite{ge2026age} studied sensing-and-then-transmit ISAC networks and introduced the age of correct sensing, with a focus on optimising the sensing-communication power split. Bai \emph{et al.}~\cite{bai2025age} explored timeliness in drone-assisted ISAC systems by jointly optimising trajectory and beamforming using deep reinforcement learning. Other related studies have considered freshness minimisation under jamming attacks in air-ground ISAC networks~\cite{mei2025aoi}, freshness-energy tradeoffs in drone-assisted ISAC networks with mobile-edge computing~\cite{hu2026balancing}, and joint sensing and freshness optimisation in energy-constrained drone-assisted ISAC systems~\cite{liu2025joint}. In contrast to these works, we study ISAC systems where sensing and communication are semantically intertwined rather than serving separate targets and users. In existing AoI-aware ISAC formulations, sensing tasks generate information about sensing targets, while communication tasks deliver data to communication users; the two functions share physical resources but the sensed information is not itself the object that must later be communicated to a decision-making entity. In our model, by contrast, sensing and communication form two consecutive stages of the same information-update process. This semantic coupling induces a two-dimensional freshness state and a non-trivial sequential decision problem.

\subsection{Contributions and Outline}
In this paper, we study timely information acquisition in ISAC systems to enhance situational awareness of a remote source. The central question is how an ISAC-enabled base station should decide between sensing new state information and communicating previously acquired information, so as to maintain fresh situational awareness under reliability and resource constraints. We formulate this problem as an infinite-horizon Markov decision process (MDP) with two-dimensional AoI states that capture information freshness at the source and at the base station. We show that the optimal stationary policy admits a monotone threshold structure characterised by a nondecreasing switching curve in the AoI state space. This yields a clear separation between sensing and communication regions and admits the operational interpretation that, as the base-station information becomes staler, the system becomes increasingly inclined to sensing rather than communication. The remainder of the paper is organised as follows. Sections~\ref{sec:single-user-problem} and \ref{sec:single-user-results} formulate the problem and present the main results. Section~\ref{sec:numerical_results} provides numerical results. Finally, Section~\ref{sec:conclusion} concludes the paper and outlines directions for future research.

\begin{figure}[t]
\centering
  \includegraphics[width=1\linewidth]{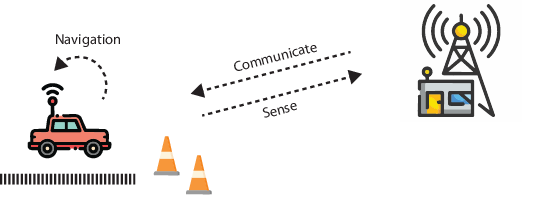}
  \caption{Illustration of an unmanned vehicle leveraging ISAC support from a base station to enhance situational awareness when onboard sensing is limited.}
  \label{fig:sketch}
\end{figure}

\section{Problem Statement}\label{sec:single-user-problem}
We consider a networked system composed of a remote source and an ISAC-enabled base station (see Fig.~\ref{fig:sketch} for a representative example). The remote source evolves over time, yet its limited onboard sensing may prevent it from directly observing its navigational state or obtaining a full overview of the surrounding environment. The base station performs remote sensing of the source state and communicates sensory information back to the source. The ISAC system at each time selects either sensing or communication mode. The goal is to minimise a staleness cost at the remote source together with the sensing and communication operational costs.

\subsection{System Model}
Let $Z_k$ denote the state of the source at time $k$ and $u_k \in \{\sense,\comm\}$ denote the ISAC action at time~$k$. Activating sensing or communication at each time incurs a fixed cost $c_s \ge 0$ or $c_c \ge 0$, respectively. When $u_k = \sense$, the ISAC system attempts to sense the current state of the source  $Z_k$, and obtains a measurement $X_k$. The sensing link is lossy and modelled as an i.i.d.\ Bernoulli process with success probability $\lambda_s \in (0,1]$. Accordingly, the measurement $X_k$ at time $k$ satisfies
\begin{align}
&\Pr(X_k = Z_k \mid Z_k, u_k = \sense) = \lambda_s,\\[1.0\jot]
&\Pr(X_k = \varnothing \mid Z_k, u_k = \sense) = 1 - \lambda_s,
\end{align}
where $\varnothing$ represents absence of data.
The base station stores the most recent successfully sensed state. Let $\tilde Z_k$ denote this stored state information at time~$k$. When $u_k = \comm$, the ISAC system transmits this information to the source. The communication link is also lossy and modelled as an i.i.d.\ Bernoulli process with success probability $\lambda_c \in (0,1]$. Let the channel input at time $k$ be $\tilde{Z}_{k}$. Accordingly, the channel output $Y_k$ at time~$k$~satisfies
\begin{align}
&\Pr(Y_k = \tilde{Z}_{k} \mid \tilde{Z}_{k}, u_k = \comm) = \lambda_c,\\[1.0\jot]
&\Pr(Y_k = \varnothing \mid \tilde{Z}_{k}, u_k = \comm) = 1 - \lambda_c,
\end{align}
where $\varnothing$ represents absence of data.

Let $\omega_k\in\{\success,\fail\}$ denote a link success indicator at time $k$. We can write
\begin{align}
\Pr(\omega_k = \success \mid u_k) =
\begin{cases}
\lambda_s, & \text{if } u_k=\sense,\\[1.0\jot]
\lambda_c, & \text{if } u_k=\comm,
\end{cases}
\end{align}
and $\Pr(\omega_k = \fail \mid u_k)=1-\Pr(\omega_k = \success \mid u_k)$. We assume that $\lambda_c \geq \lambda_s$, which captures the practical fact that communication links typically benefit from error-control mechanisms. We also assume that $c_c \geq c_s$, meaning that the cost of communication is at least equal to that of sensing. Let $\hat Z_k$ denote the state information available at the source at time $k$ following the ISAC action. The update rule for $\hat Z_k$ at the source is then given by
\begin{equation}
\hat Z_{k} =
\begin{cases}
\tilde{Z}_{k}, & \text{if } u_k = \comm \text{ and } \omega_k = \success,\\[1.0\jot]
\hat Z_{k-1}, & \text{otherwise}.
\end{cases}\label{eq:remote_update}
\end{equation}
This implies that the state information available at the source is updated only upon a successful communication action, while the source receives no data under a sensing action and hence its state information remains unchanged.

\subsection{Freshness Metric}
We define the age of information (AoI) at an entity, namely the source or the base station, as the elapsed time since the most recent state information available at that entity was successfully acquired. For $e \in \{s,b\}$, where $s$ and $b$ index the source and the base station, let $\alpha^e_{k}$ denote the AoI at time $k$ before the ISAC action, and $\alpha^e_{k^+}$ denote the AoI at time $k$ after the ISAC action. These variables quantify the freshness of the state information $\hat{Z}_k$ and $\tilde{Z}_k$. The evolution of the AoI variables at the source and the base station is governed by the outcomes of sensing and communication actions. In particular, under a successful sensing action at time $k$, the base station directly refreshes its local state information, while the source is not updated. Hence, we have $\alpha^s_{k^+} = \alpha^s_k$ and $\alpha^b_{k^+} = 0$. Moreover, under a successful communication action at time $k$, the base station conveys its local state information to the source, so that the two sides become aligned, while the base station itself is not updated. Hence, we have $\alpha^s_{k^+} = \alpha^b_k$ and $\alpha^b_{k^+} = \alpha^b_k$. However, if neither sensing nor communication action succeeds, no entity is updated. Therefore, we have $\alpha^e_{k^+} = \alpha^e_{k}$. Finally, between time instants $k^+$ and $k+1$, the AoI increases linearly with time, and we have $\alpha^e_{k+1} = \alpha^e_{k^+} + 1$.

\subsection{Decision-Making Problem}
We seek to find an optimal policy for the ISAC system. Throughout, we adopt the convention $u_k \in \{0,1\} \equiv \{\sense,\comm\}$ and $\omega_k \in \{0,1\} \equiv \{\fail,\success\}$. Let the state of the status-updating system at time $k$ be $S_k = (\alpha^s_k, \alpha^b_k)$. For simplicity of exposition, when the time index is immaterial, we use $i$ and $j$ instead of $\alpha^s$ and $\alpha^b$. Let $H_k := (S_0,u_0,\omega_0,S_1,u_1,\omega_1,\ldots,S_k)$ denote the history observed up to decision time \(k\). A policy for the ISAC system is a sequence $\pi := \{\mu_k\}$, where each decision rule \(\mu_k\) maps the history $H_k$ to an ISAC action $u_k$, i.e., $u_k = \mu_k(H_k)$. The class of all admissible policies is denoted by \(\mathcal P\). The stage cost measures the instantaneous penalty associated with stale information at the source, together with the resource cost of the selected action, and is defined as
\begin{align*}
g(S_k, u_k) = \alpha^s_k + c_s \mathds{1} \big\{ u_k=0 \big \} + c_c \mathds{1} \big \{ u_k=1 \big \}.
\end{align*}
We consider a long-run discounted-cost formulation to capture a balance between immediate and future costs in the system. In particular, we would like to solve
\begin{align}
\underset{\pi \in \mathcal{P}}{\minimize} \quad J(\pi) := \E^\pi \bigg[\sum_{k=0}^{\infty} \gamma^k\,g\big(S_k,u_k\big)\bigg], \label{eq:main_problem}
\end{align}
where $\gamma \in (0,1)$ is the discount factor. An optimal policy is denoted by $\pi^\star$.

\section{Main Results}\label{sec:single-user-results}
In this section, we present our main results. Note that the problem in (\ref{eq:main_problem}) can be cast as an infinite-horizon MDP with state $S = (\alpha^s, \alpha^b)$, which follows from the fact that the AoI pair constitutes a sufficient state variable for the status-updating system: the system evolution and the stage cost depend only on the current state and the current action, and not on the past history. Let $\mathcal{S}$ denote the state space and $\mathcal{U}$ denote the action space. As stated, at each time $k$, the action $u_k\in\{0,1\}$ determines whether a sensing or communication operation is attempted. Moreover, the link success indicator $\omega_k \in\{0,1\}$ specifies whether an operation is successful or failed. Depending on $u_k$ and $\omega_k$, the AoI pair evolves as
\begin{align*}
&(\alpha^s_{k+1}, \alpha^b_{k+1}) =
\begin{cases}
(\alpha^s_k + 1, 1), & \text{if } (u_k,\omega_k)=(0,1),\\[1.0\jot]
(\alpha^s_k + 1, \alpha^b_k + 1), & \text{if } (u_k,\omega_k)=(0,0),\\[1.0\jot]
(\alpha^b_k+1, \alpha^b_k + 1), & \text{if } (u_k,\omega_k)=(1,1),\\[1.0\jot]
(\alpha^s_k + 1, \alpha^b_k + 1), & \text{if } (u_k,\omega_k)=(1,0).
\end{cases}
\end{align*}
Note that all transitions are expressed at decision epochs $k+1$, i.e., after the post-action increment. The reachable state space is $\mathcal S:=\{(i,j)\in\mathbb N^2:1\le j\le i\}$. 

Let $F_{u,\omega}(S)$ be the transition mapping, giving the next state that results given action $u$ in state $S$ and outcome $\omega$. For any function $V : S \to \mathbb{R}$, define the action-value function under action $u \in \{\sense,\comm\}$ as
\begin{align}
  Q^V_u(S) := g(S, u) + \gamma \, \E\Big[ V\big(F_{u,\omega}(S)\big) \Big], \label{eq:def:Qu}
\end{align}
where the expectation is taken with respect to $\omega$. The Bellman operator can be written as
\begin{align}
TV(S) := \min_{u \in \{\sense,\comm\}} Q^V_u(S). \label{eq:def:BellmanOp}
\end{align}
The action-value functions can be written as
\begin{align*}
Q^V_{\sense}(i,j) &= i + c_s + \gamma \lambda_s V(i+1,1) \\[1.0\jot]
&\quad + \gamma(1-\lambda_s) V(i+1,j+1), \\[1.0\jot]
Q^V_{\comm}(i,j) &= i + c_c + \gamma \lambda_c V(j+1,j+1) \\[1.0\jot]
&\quad + \gamma(1-\lambda_c)V(i+1,j+1).
\end{align*}

Following the standard weighted sup-norm contraction result for discounted Markov decision processes with unbounded costs, one can establish well-posedness of the discounted Bellman equation that guarantees existence and uniqueness of the value function. Furthermore, one can show that the value function satisfies a marginal dominance property, i.e.,
\begin{align*}
V^{\star}(i,j+1)-V^{\star}(i,j) \le V^{\star}(j+1,j+1)-V^{\star}(j,j),
\end{align*}
for all $i \geq j+1$. Given these observations, we next characterise single-crossing properties.
\begin{lemma}\label{lem:single-crossing}
Define the difference function $\Delta(i,j):=Q_{\sense}^{V^{\star}}(i,j)-Q_{\comm}^{V^{\star}}(i,j)$.
Then:
\begin{enumerate}
\item For every fixed $j$, $i\mapsto \Delta(i,j)$ is nondecreasing, and $\{i\ge j:\Delta(i,j)\le 0\}$ is a lower set in $i$.
\item For every fixed $i$, $j\mapsto \Delta(i,j)$ is nonincreasing, and $\{j\le i:\Delta(i,j)\le 0\}$ is an upper set in $j$.
\end{enumerate}
\end{lemma}

\begin{IEEEproof}
Using the definitions of \(Q_{\sense}^{V^{\star}}\) and \(Q_{\comm}^{V^{\star}}\), the difference function is obtained as
\begin{align}
\Delta(i,j) &= (c_s-c_c)\!+\!\gamma\lambda_s V^{\star}(i+1,1)\!-\!\gamma\lambda_c V^{\star}(j+1,j+1) \notag\\[1.0\jot]
&\quad  +\gamma(\lambda_c-\lambda_s)V^{\star}(i+1,j+1). \label{eq:delta-expanded}
\end{align}

\noindent \textit{Step 1: Monotonicity in $i$.} Fix $j \geq 1$. We show that $\Delta(i+1,j) \geq \Delta(i,j)$ for any $i \geq j$. Using \eqref{eq:delta-expanded}, we obtain
\begin{align*}
&\Delta(i+1,j)-\Delta(i,j)\\[1.0\jot]
&= \gamma\lambda_s \big[ V^{\star}(i+2,1)-V^{\star}(i+1,1) \big] \notag\\[1.0\jot]
&\quad +\gamma(\lambda_c-\lambda_s) \big[ V^{\star}(i+2,j+1)-V^{\star}(i+1,j+1) \big].
\end{align*}
Since $V^{\star}$ is coordinatewise nondecreasing, we have
\begin{align*}
V^{\star}(i+2,1)-V^{\star}(i+1,1)\ge0,\\[1.0\jot]
V^{\star}(i+2,j+1)-V^{\star}(i+1,j+1)\ge0.
\end{align*}
As $\gamma > 0$, $\lambda_s>0$, and $\lambda_c-\lambda_s \geq 0$, we conclude that $\Delta(i+1,j)-\Delta(i,j)\ge0$. This implies that, for every fixed $j$, the mapping $i\mapsto \Delta(i,j)$ is nondecreasing.

\noindent \textit{Step 2: Lower-set property in $i$.} Fix $j \geq 1$. Define the sensing region $\mathcal S_j:=\{i\ge j:\Delta(i,j) \leq 0\}$. Let $i_2\in \mathcal S_j$, and suppose $j \leq i_1 \leq i_2$. Since $i\mapsto \Delta(i,j)$ is nondecreasing, we find that $\Delta(i_1,j)\le \Delta(i_2,j)\le0$. Hence, $i_1\in\mathcal S_j$ and $\mathcal S_j$ is a lower~set.

\noindent \textit{Step 3: Monotonicity in \(j\).} Fix $i \geq 1$. We show that $\Delta(i,j+1)-\Delta(i,j) \leq 0$ for any $j \leq i-1$. Using \eqref{eq:delta-expanded}, we obtain
\begin{align*}
&\Delta(i,j+1)-\Delta(i,j) \\[1.0\jot]
&= -\gamma\lambda_c \big[ V^{\star}(j+2,j+2)-V^{\star}(j+1,j+1) \big] \\[1.0\jot]
&\quad +\gamma(\lambda_c-\lambda_s) \big[ V^{\star}(i+1,j+2)-V^{\star}(i+1,j+1) \big]\\[1.0\jot]
&= -\gamma\lambda_c D_d V^{\star}(j+1) + \gamma(\lambda_c-\lambda_s)D_vV^{\star}(i+1,j+1),
\end{align*}
where
\begin{align*}
D_d V^{\star}(j) &:= V^{\star}(j+1,j+1) - V^{\star}(j,j),\\[1.0\jot]
D_v V^{\star}(i,j) &:= V^{\star}(i,j+1) - V^{\star}(i,j).
\end{align*}
By the marginal dominance property, as $j \leq i-1$, we have
\[
D_v V^{\star}(i+1,j+1)\le D_d V^{\star}(j+1).
\]
As $\gamma>0$ and $\lambda_c-\lambda_s\ge0$, we obtain
\begin{align}
&\Delta(i,j+1)-\Delta(i,j) \notag\\[1.0\jot]
&\le -\gamma\lambda_c D_d V^{\star}(j+1) +\gamma(\lambda_c-\lambda_s)D_d V^{\star}(j+1) \notag\\[1.0\jot]
&= -\gamma\lambda_s D_d V^{\star}(j+1).
\end{align}
Since the stage cost includes the term \(\alpha_k^s\), increasing the diagonal state from \((j+1,j+1)\) to \((j+2,j+2)\) increases the immediate cost by at least one, and therefore $D_dV^\star(j+1) := V^\star(j+2,j+2)-V^\star(j+1,j+1) \ge 1$. As $\lambda_s>0$, we conclude that $\Delta(i,j+1)-\Delta(i,j)\le0$. This implies that, for every fixed $i$, the mapping $j\mapsto \Delta(i,j)$ is nonincreasing.

\noindent \textit{Step 4: Upper-set property in \(j\).} Fix $i \geq 1$. Define the sensing region $\mathcal S_i:=\{j\le i:\Delta(i,j) \leq 0\}$. Let \(j_1\in\mathcal S_i\), and suppose \(j_1\le j_2\le i\). Since \(j\mapsto \Delta(i,j)\) is nonincreasing, $\Delta(i,j_2)\le \Delta(i,j_1)\le0$. Hence, $j_2\in\mathcal S_i$, and $\mathcal S_i$ is an upper set. This completes the proof of the lemma.
\end{IEEEproof}

We now use the single-crossing properties to prove the threshold structure of the optimal policy.
\begin{theorem}\label{thm:threshold_policy}
For every fixed $j$, there exists $\tau(j)\in \mathbb N\cup\{+\infty\}$ such that
\begin{align*}
\pi^{\star}(i,j) =
\begin{cases}
\mathrm{sense}, & i \le \tau(j),\\[1.0\jot]
\mathrm{comm}, & i > \tau(j).
\end{cases}
\end{align*}
Moreover, the threshold \(\tau(j)\) is nondecreasing in \(j\).
\end{theorem}

\begin{IEEEproof}
Consider a tie-breaking convention such that sensing is selected when $Q_{\sense}^{V^{\star}}=Q_{\comm}^{V^{\star}}$. Under this convention, sensing is optimal if and only if $\Delta(i,j)\le0$.

\noindent \textit{Step 1: Threshold structure of the optimal policy.} Fix \(j\ge1\). Consider the sensing region $\mathcal S_j = \{i\ge j: \Delta(i,j)\le0\}$. By Lemma~\ref{lem:single-crossing}, \(\mathcal S_j\) is a lower set in \(\{i\ge j\}\). That is, if \(i_2\in\mathcal S_j\) and \(j\le i_1\le i_2\), then $\Delta(i_1,j)\le \Delta(i_2,j)\le0$. Under \(c_c\ge c_s\), sensing is optimal on every diagonal state \((j,j)\). Moreover, by diagonal sensing, $\Delta(j,j)\le0$. Thus \(j\in\mathcal S_j\), so \(\mathcal S_j\) is nonempty. Define $\tau(j):=\sup \mathcal S_j$. Since \(\mathcal S_j\) is a nonempty lower set of integers, its supremum is either a largest integer in \(\mathcal S_j\) or \(+\infty\). Therefore, the sensing region can be expressed as $\mathcal S_j=\{i\ge j: i\le\tau(j)\}$. Consequently, $\pi^{\star}(i,j)=\mathrm{sense}$ if $i\le \tau(j)$, and $\pi^{\star}(i,j)=\mathrm{comm}$ otherwise. The case $\tau(j)=+\infty$ corresponds to sensing being optimal for all states in row~$j$.

\noindent \textit{Step 2: Monotonicity of threshold in $j$.} It remains to prove that \(\tau(j)\) is nondecreasing in \(j\). Fix \(j_1\le j_2\). We first prove the inclusion $\mathcal S_{j_1}\cap\{i\ge j_2\} \subseteq \mathcal S_{j_2}$. Indeed, let \(i\in\mathcal S_{j_1}\cap\{i\ge j_2\}\). Then, $\Delta(i,j_1)\le0$. By Lemma~\ref{lem:single-crossing}, \(j\mapsto\Delta(i,j)\) is nonincreasing in $j$. Since \(j_1\le j_2\), we have $\Delta(i,j_2)\le \Delta(i,j_1)\le0$. Hence, \(i\in\mathcal S_{j_2}\), proving the inclusion. Now, suppose that $\tau(j_2)<\tau(j_1)$. Because \(\mathcal S_{j_1}\) is a lower set of integers with supremum \(\tau(j_1)\), there exists $i\in\mathcal S_{j_1}$ such that $i>\tau(j_2)$. Furthermore, diagonal sensing gives \(j_2\in\mathcal S_{j_2}\), and therefore $\tau(j_2)\ge j_2$. Thus, $i>\tau(j_2)\ge j_2$, so $i\ge j_2$. Hence, $i\in\mathcal S_{j_1}\cap\{i\ge j_2\}$. By the inclusion proved above, $i\in\mathcal S_{j_2}$. But this contradicts \(i>\tau(j_2)\), since no element of \(\mathcal S_{j_2}\) can exceed its supremum. Therefore, we conclude that $\tau(j_1)\le\tau(j_2)$ and the threshold \(\tau(j)\) is nondecreasing in \(j\).
\end{IEEEproof}

\section{Numerical Analysis}\label{sec:numerical_results}
In this section, we numerically examine our analytical results. The numerical experiments illustrate the value-function structure and the optimal threshold policy. In our simulations, the state space is $\mathcal S_A$, and value iteration is run until the Bellman residual falls below a prescribed tolerance. The truncation level is chosen sufficiently large so that the switching boundary is unaffected by boundary effects. Unless otherwise stated, the following baseline parameters are used in the simulations: $A_{\max}=30$, $\gamma=0.95$, $\lambda_s=0.30$, $\lambda_c=0.35$, $c_s=0.10$, and $c_c=0.20$. We examine how sensing reliability, communication reliability, and action costs shape the optimal policy. Fig.~\ref{fig:value_function} shows the value function $V^\star(\alpha^s,\alpha^b)$ obtained from value iteration. The value function is coordinatewise nondecreasing, in agreement with the analytical structure of the value function. The dependence on $\alpha^s$ is visibly steeper than on $\alpha^b$, as the stage cost penalises the source AoI directly while the base-station AoI enters the cost only indirectly through future transitions. Fig.~\ref{fig:decision_map} shows the corresponding optimal policy. The sensing and communication regions are separated by a clear switching boundary: for each fixed $\alpha^b$, sensing is optimal when $\alpha^s$ lies below a threshold, and communication becomes optimal once the threshold is crossed. The threshold is nondecreasing in $\alpha^b$, confirming the analytical structure of the threshold: as the base-station information becomes staler, the system tolerates larger source AoI before switching to communication.

\begin{figure}[t!]
\centering
\includegraphics[width=0.85\linewidth, trim={0mm 0mm 0mm 5mm},clip]{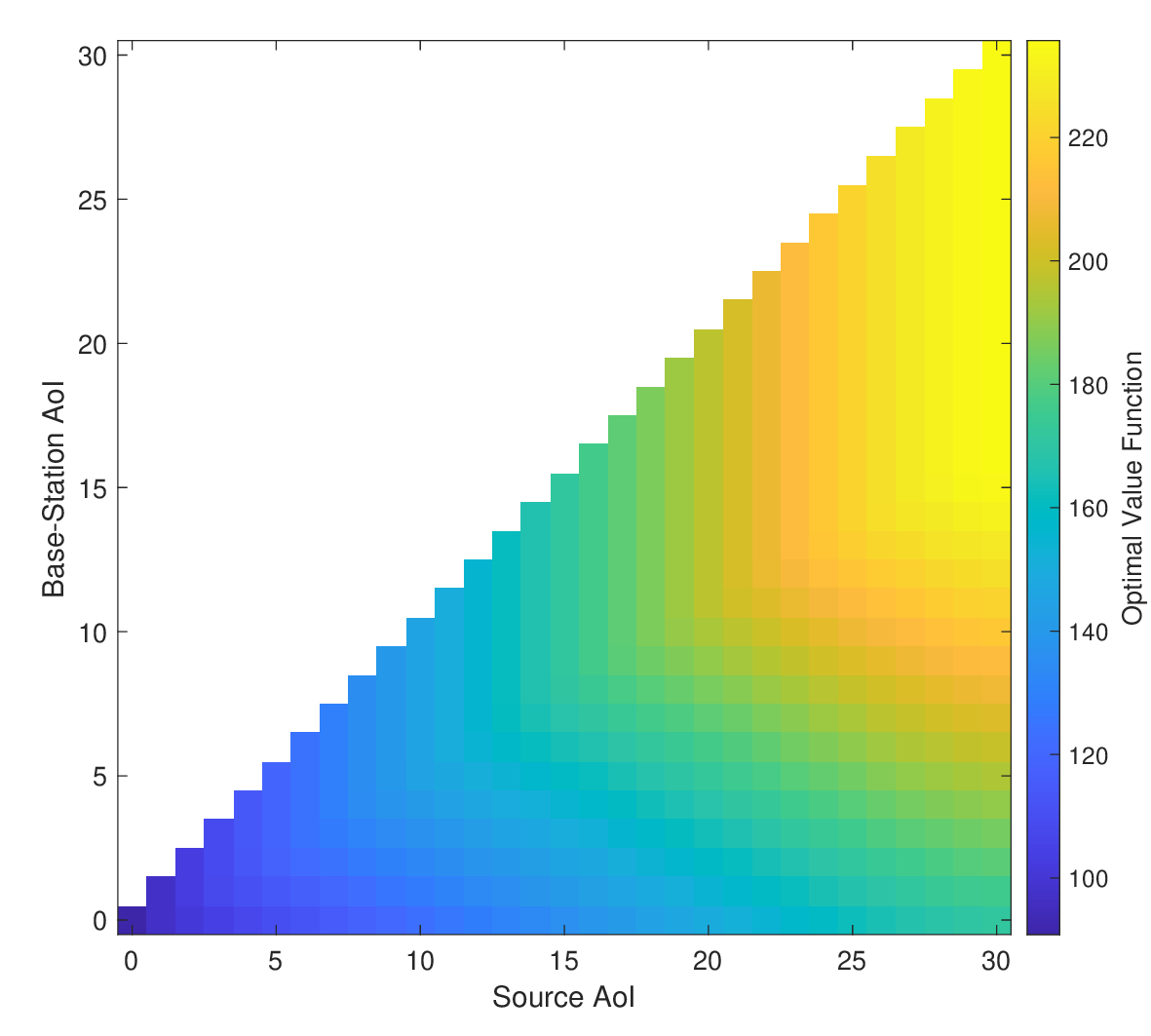}
\caption{Optimal value function, showing how the expected cost varies with the source and base-station AoI values.}
\label{fig:value_function}
\end{figure}

\begin{figure}[t!]
\centering
\includegraphics[width=0.85\linewidth, trim={0mm 0mm 0mm 5mm},clip]{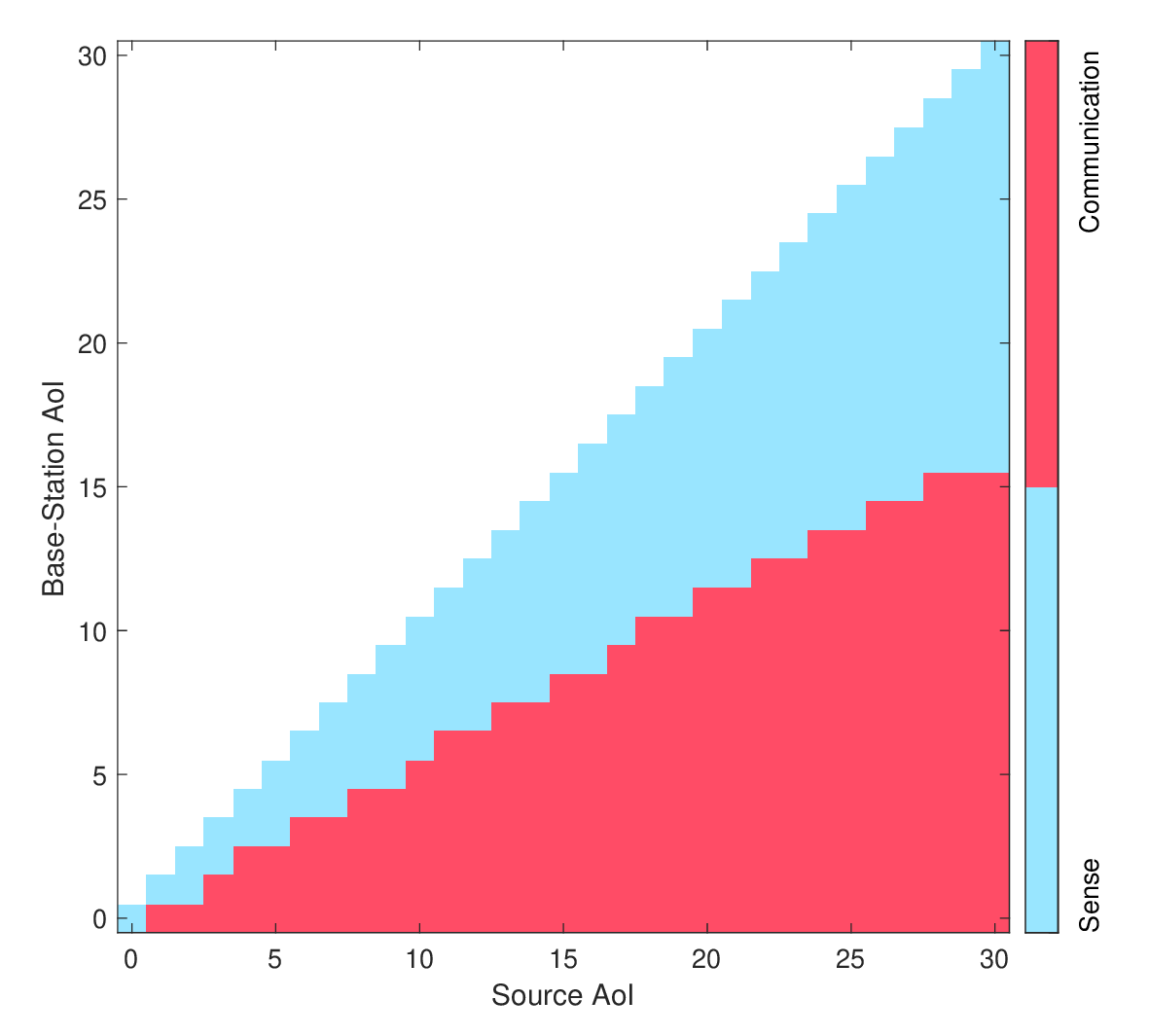}
\caption{Optimal decision map, illustrating the sensing and communication regions induced by the optimal threshold policy.}
\label{fig:decision_map}
\end{figure}

\section{Conclusions}\label{sec:conclusion}
In this paper, we studied timely information acquisition in ISAC systems to enhance situational awareness of a remote source. We showed that the optimal discounted policy admits a monotone threshold structure. This result was obtained based on a marginal-dominance property that yields the required single-crossing behaviour. The numerical results validated the analytical results. Future work may extend the proposed framework by incorporating explicit source dynamics, explicit control performance criteria, and joint physical-layer optimisation through waveform synthesis, signalling schemes, beamforming, power control, and interference management.

\bibliography{../../../mybib}
\bibliographystyle{ieeetr}

\end{document}